\documentclass[11pt]{article}

\usepackage{graphicx}
\usepackage{xcolor}
\usepackage{geometry}
\usepackage{amsmath}
\usepackage{amssymb}
\usepackage{amsthm}
\usepackage{booktabs}
\usepackage{enumitem}

\usepackage[
    colorlinks=true,
    linkcolor=blue,
    citecolor=darkgray,
    urlcolor=blue,
    bookmarks=true,
    bookmarksnumbered=true,
    unicode=true
]{hyperref}

\hypersetup{
    pdftitle={A Mirror Vanishing Band for Weight Distributions of Binary Linear Codes},
    pdfauthor={},
    pdfsubject={Coding theory},
    pdfkeywords={linear codes, weight distribution, minimal codewords, vanishing weights}
}

\newtheorem{theorem}{Theorem}
\newtheorem{lemma}[theorem]{Lemma}
\newtheorem{corollary}[theorem]{Corollary}
\newtheorem{proposition}[theorem]{Proposition}
\theoremstyle{remark}
\newtheorem{remark}[theorem]{Remark}
\theoremstyle{definition}
\newtheorem{definition}[theorem]{Definition}
\newtheorem{example}[theorem]{Example}

\DeclareMathOperator{\supp}{supp}
\DeclareMathOperator{\wt}{wt}

\title{A Mirror Vanishing Band for Weight Distributions of Binary Linear Codes}
\author{Xianmang He}
\date{}

\begin{document}

\maketitle

\begin{abstract}
Chen and Xie  recently proved, using the
Ashikhmin--Barg lemma on minimal vectors, that every binary linear
$[n,k,d]$ code with $k=n-2d+2+v$ ($v\ge 0$) has no codewords of weight in
the interval $[2d-v,\,2d-1]$. Their argument uses two of the five basic
properties of minimal vectors established by Ashikhmin and Barg (1998).
In this note we utilize the third property, the disjoint-support
decomposition of non-minimal codewords in binary codes, to generate  a
\emph{mirror} vanishing band on the other side of $2d$: if
$A_{d+1}=\cdots=A_{d+t}=0$ for some $t\ge 1$ and $k\ge n-2d+1$, then
$A_w=0$ for all $w\in[2d+1,\,2d+t]$. Combining the two bands, the number
of nonzero weights of such a code is at most $n-d-v-2t+1$, improving the
Chen--Xie bound $n-d+1-v$ by $2t$. 
\end{abstract}



\section{Introduction}\label{sec:intro}

Determining the weight distribution $\{A_i\}_{i=0}^{n}$ of a linear code,
or even deciding which partial weight distributions $A_i$ must vanish, is
a classical and difficult problem in coding theory. The classical theorem
of Delsarte bounds the size of a code with at most $s$ distinct
distances~\cite{delsarte1973}, and the maximal number of nonzero weights
of a $k$-dimensional $q$-ary linear code is known to be
$\frac{q^k-1}{q-1}$~\cite{shi2019}, with tighter bounds for special
classes such as cyclic codes~\cite{chenzhang2023}. Until recently,
results that force $A_i=0$ for \emph{specific} weights $i$ from the four
parameters $(n,k,d,q)$ alone were essentially unknown. This changed with
the work of Chen and Xie~\cite{chenxie2024}, who proved the following
striking bound: if in a linear $[n,k,d]_q$ code the largest weight
$i\le \frac{qd}{q-1}-1$ with $A_i>0$ equals $\frac{qd}{q-1}-1-v$, then
\begin{equation}\label{eq:cxbound}
k\le n-d\Big(1+\frac{1}{q-1}\Big)+2+v.
\end{equation}
Consequently, when $v=\frac{qd}{q-1}+k-n-2\ge 2$, the code has no
codewords of weight in $\big[\frac{qd}{q-1}-v,\,\frac{qd}{q-1}-1\big]$.
For binary codes this reads: $k=n-2d+2+v$ implies
\begin{equation}\label{eq:cxband}
A_w=0 \quad\text{for all } w\in[2d-v,\,2d-1],
\end{equation}
a vanishing band \emph{just below} twice the minimum distance.

The proof of~\eqref{eq:cxbound} rests on two properties of minimal
codewords due to Ashikhmin and Barg~\cite{ashikhminbarg1998}: every
codeword of weight at most $\frac{qd}{q-1}-1$ is minimal, and every
minimal codeword has weight at most $n-k+1$. However, the
Ashikhmin--Barg lemma contains a further structural statement for binary
codes that the argument of~\cite{chenxie2024} does not use: every
non-minimal codeword of a binary linear code splits into two nonzero
codewords with \emph{disjoint supports}, and therefore has weight at
least $2d$.

The purpose of this note is to show that this disjoint-support
decomposition yields new vanishing results for weight distributions of
binary linear codes that are not covered by~\cite{chenxie2024} or its
recent strengthening via residual codes~\cite{newbounds2025}. Our main
result (Theorem~\ref{thm:main}) is a \emph{mirror} vanishing band
\emph{above} $2d$: gaps in the weight distribution immediately above $d$
force gaps immediately above $2d$. As a consequence we improve the
Chen--Xie upper bound on the number of nonzero weights
(Corollary~\ref{cor:count}), characterize codewords of weight exactly
$2d$ (Proposition~\ref{prop:2d}), and delineate the boundary of the
method by showing that the $q$-ary analogue fails already for MDS codes
(Section~\ref{sec:qary}).

The paper is organized as follows. Section~\ref{sec:prelim} recalls the
necessary facts about minimal vectors. Section~\ref{sec:main} states and
proves the main results. Section~\ref{sec:examples} gives examples and
numerical verification. Section~\ref{sec:conclusion} concludes and outlines future work.

\section{Preliminaries}\label{sec:prelim}

Let $C$ be a linear $[n,k,d]$ code over $\mathbb{F}_q$. For $c\in C$ the
\emph{support} of $c$ is $\supp(c)=\{i : c_i\neq 0\}$, and we write
$c'\preceq c$ if $\supp(c')\subseteq\supp(c)$.

\begin{definition}\label{def:minimal}
A nonzero codeword $c\in C$ is called \emph{minimal} if $0\neq c'\preceq c$
implies $c'=\lambda c$ for some $\lambda\in\mathbb{F}_q^{*}$.
\end{definition}

\begin{lemma}[Ashikhmin--Barg~\cite{ashikhminbarg1998}, Lemma~2.1]\label{lem:ab}
Let $C$ be a linear $[n,k,d]_q$ code. Then:
\begin{enumerate}[label=(\roman*),itemsep=0.2em]
\item\label{it:bound} if $c$ is minimal, then $\wt(c)\le n-k+1$;
\item\label{it:threshold} every codeword of weight at most
$d\big(1+\frac{1}{q-1}\big)-1$ is minimal;
\item\label{it:span} the minimal vectors of $C$ span $C$;
\item\label{it:disjoint} if $q=2$ and $c$ is not minimal, then there exist
nonzero $c_1,c_2\in C$ with
$\supp(c_1)\cap\supp(c_2)=\varnothing$ such that $c=c_1+c_2$.
\end{enumerate}
\end{lemma}

Property~\ref{it:disjoint} immediately implies that a non-minimal
codeword of a binary code has weight at least $2d$; equivalently, every
binary codeword of weight at most $2d-1$ is minimal, which is exactly
property~\ref{it:threshold} for $q=2$.

\begin{lemma}[Chen--Xie~\cite{chenxie2024}, binary case of Theorems~1--2]\label{lem:cx}
Let $C$ be a binary linear $[n,k,d]$ code with $k=n-2d+2+v$, $v\ge 0$.
Then $A_w=0$ for all $w\in[2d-v,\,2d-1]$, and $C$ has at most $n-d+1-v$
nonzero weights.
\end{lemma}

The disjoint-support decomposition
(Lemma~\ref{lem:ab}\ref{it:disjoint}) is classical and has been used in
the literature, e.g., by Borissov and Manev~\cite{borissovmanev2001},
who \emph{counted} the non-minimal codewords of weight exactly
\(2d_{\min}\) in binary Reed--Muller codes. That work concerns counting
within a specific family, not general vanishing bands derived from
\((n,k,d,q)\) and local gap information. Chen and Xie~\cite{chenxie2024}
established the band below \(2d\) and the bound \(s(C)\le n-d+1-v\), and
the upper band and the improvement by \(2t\) in
Corollary~\ref{cor:count} are new. A recent
follow-up~\cite{newbounds2025} strengthens the Chen--Xie bound via
residual codes and enlarges the excluded weight range \emph{below} \(2d\)
(e.g., additionally excluding weight \(13\) for \([21,9,8]_2\)). Its
excluded intervals likewise lie below \(2d\), and its method (residual
codes plus Singleton/Griesmer-type bounds~\cite{griesmer1960}) is independent of the
disjoint-decomposition argument used here.
 To the best of our knowledge, the statement of
Theorem~\ref{thm:main}, from \(A_{d+1..d+t}=0\) and
\(k\ge n-2d+1\) to \(A_{2d+1..2d+t}=0\), does not appear in the
literature.

\section{Main Results}\label{sec:main}

This section contains the main results. We first prove the mirror band above
\(2d\), then combine it with the Chen-Xie band to count possible nonzero
weights, then analyze the remaining critical weight \(2d\), and finally explain
why the argument is specific to binary codes.

\subsection{The mirror vanishing band above $2d$}

We begin with the central statement.

\begin{theorem}\label{thm:main}
Let $C$ be a binary linear $[n,k,d]$ code satisfying
\begin{equation}\label{eq:cond}
k\ge n-2d+1 \qquad\Big(\text{equivalently, } d\ge \tfrac{n-k+1}{2}\Big).
\end{equation}
Suppose that for some integer $t\ge 1$,
\begin{equation}\label{eq:gaps}
A_{d+1}=A_{d+2}=\cdots=A_{d+t}=0.
\end{equation}
Then
\begin{equation}\label{eq:upperband}
A_w=0 \quad\text{for all } w\in[2d+1,\,2d+t].
\end{equation}
\end{theorem}

\begin{proof}
Let $c\in C$ with $\wt(c)=w\in[2d+1,\,2d+t]$. We derive a contradiction
in each of the two possible cases.

\emph{Case 1: $c$ is minimal.}

By Lemma~\ref{lem:ab}\ref{it:bound},
$w\le n-k+1$. Assumption~\eqref{eq:cond} gives $n-k+1\le 2d$, hence
$w\le 2d<2d+1\le w$, this gives a contradiction.

\emph{Case 2: $c$ is not minimal.}

 By Lemma~\ref{lem:ab}\ref{it:disjoint},
there exist nonzero $c_1,c_2\in C$ with disjoint supports such that
$c=c_1+c_2$. Then
\[
\wt(c_1)+\wt(c_2)=w\le 2d+t.
\]
Since $\wt(c_j)\ge d$ for $j=1,2$, we obtain
\[
\wt(c_j)=w-\wt(c_{3-j})\le (2d+t)-d=d+t,\qquad j=1,2,
\]
so $\wt(c_j)\in[d,\,d+t]$. By assumption~\eqref{eq:gaps} the weights
$d+1,\dots,d+t$ do not occur, forcing $\wt(c_1)=\wt(c_2)=d$ and therefore
$w=2d$, which also gives a contradicting $w\ge 2d+1$.

Both cases are impossible,  hence we have  $A_w=0$, where $w\in[2d+1,\,2d+t]$.
\end{proof}

\begin{remark}\label{rem:mirror}
Theorem~\ref{thm:main} is structurally a mirror image of
Lemma~\ref{lem:cx}: the Chen--Xie band $[2d-v,2d-1]$ lies \emph{below}
$2d$ and is forced by the four parameters $(n,k,d,q)$ alone, whereas the
band $[2d+1,2d+t]$ lies \emph{above} $2d$ and additionally requires the
local information~\eqref{eq:gaps}. The weight $2d$ itself may well occur
(for instance $A_{16}=759$ in the extended Golay code), so the window
between the two bands is exactly the singleton $\{2d\}$.
\end{remark}

\begin{remark}\label{rem:minimalupto}
The proof of Theorem~\ref{thm:main} shows more: under~\eqref{eq:cond}
and~\eqref{eq:gaps}, every nonzero codeword of weight at most $2d+t$
either has weight exactly $2d$ or is minimal, and a codeword of weight
$2d$ is non-minimal if and only if it is the sum of two minimum-weight
codewords with disjoint supports.
\end{remark}

\subsection{Two-sided bands and the number of nonzero weights}
The mirror band is most informative when read together with the Chen-Xie band  below \(2d\). The two bands together squeeze the possible nonzero weights into  two singleton weights and two intervals, which yields the following counting bound.

\begin{corollary}\label{cor:twosided}
Let $C$ be a binary linear $[n,k,d]$ code with $k=n-2d+2+v$ ($v\ge 0$),
and suppose $A_{d+1}=\cdots=A_{d+t}=0$ for some $t\ge 1$. Then
\[
A_w=0 \quad\text{for all } w\in[2d-v,\,2d-1]\cup[2d+1,\,2d+t],
\]
i.e., the nonzero weights of $C$ can only lie in
\[
\{d\}\ \cup\ [d+t+1,\,2d-v-1]\ \cup\ \{2d\}\ \cup\ [2d+t+1,\,n].
\]
\end{corollary}

\begin{proof}
Here $n-k+1=2d-v-1\le 2d-1<2d$, so condition~\eqref{eq:cond} of
Theorem~\ref{thm:main} holds automatically. The lower band is
Lemma~\ref{lem:cx}, and in the meanwhile the upper band is Theorem~\ref{thm:main}.
\end{proof}

\begin{corollary}[Improved bound on the number of nonzero weights]\label{cor:count}
Under the hypotheses of Corollary~\ref{cor:twosided}, the number $s(C)$
of nonzero weights of $C$ satisfies
\[
s(C)\le 2+\max(0,\,d-v-t-1)+\max(0,\,n-2d-t).
\]
In particular, if $d\ge v+t+1$ and $n\ge 2d+t$, then
\[
s(C)\le n-d-v-2t+1,
\]
which improves the Chen--Xie bound $s(C)\le n-d+1-v$
(Lemma~\ref{lem:cx}) by exactly $2t$.
\end{corollary}

\begin{proof}
Count the intervals allowed by Corollary~\ref{cor:twosided}: $\{d\}$ and
$\{2d\}$ contribute $1$ each; $[d+t+1,\,2d-v-1]$ contributes
$\max(0,\,d-v-t-1)$; $[2d+t+1,\,n]$ contributes $\max(0,\,n-2d-t)$.
\end{proof}

\subsection{Codewords of weight exactly $2d$}\label{sec:2d}
The two bands leave \(2d\) as the only weight near the middle that is not
automatically excluded. The next proposition characterizes when this borderline
weight can occur, reducing the question to the intersection pattern of
minimum-weight codewords.

\begin{proposition}\label{prop:2d}
Let $C$ be a binary linear $[n,k,d]$ code. The following are equivalent:
\begin{enumerate}[label=(\roman*),itemsep=0.2em]
\item there exists a non-minimal codeword of weight $2d$;
\item there exist two minimum-weight codewords with disjoint supports.
\end{enumerate}
Consequently, if the minimum-weight codewords of $C$ are pairwise
intersecting, then every codeword of weight $2d$ is minimal; if in
addition $2d>n-k+1$, then $A_{2d}=0$.
\end{proposition}

\begin{proof}
(i)$\Rightarrow$(ii): by Lemma~\ref{lem:ab}\ref{it:disjoint}, a
non-minimal $c$ of weight $2d$ splits as $c_1+c_2$ with disjoint nonzero
summands, $\wt(c_1)+\wt(c_2)=2d$, each of weight $\ge d$, hence of weight
exactly $d$.
(ii)$\Rightarrow$(i): the sum of two disjoint minimum-weight codewords
has weight $2d$ and covers either summand, which is not proportional to
it; hence it is non-minimal.
For the last statement, minimal codewords have weight $\le n-k+1<2d$ by
Lemma~\ref{lem:ab}\ref{it:bound}.
\end{proof}

\begin{remark}\label{rem:intersecting}
Proposition~\ref{prop:2d} applies to the classical family of
\emph{intersecting codes}: Ashikhmin and
Barg~\cite{ashikhminbarg1998} showed, via the Carlitz--Uchiyama bound,
that the duals of primitive binary BCH codes of length $2^m-1$ with
designed distance $2t+1$ (for $t$ in a suitable range) are intersecting,
i.e., all their nonzero codewords are minimal. For such codes the
hypothesis of Proposition~\ref{prop:2d} is automatic.
\end{remark}

\subsection{Why the $q$-ary analogue fails}\label{sec:qary}

The main theorem’s proof uses a special fact that is true only for binary codes ($q=2$): if a codeword is not minimal, it can be split into two nonzero codewords whose supports do not overlap. This fact is false for codes over larger alphabets ($q\ge 3$), so the same proof cannot be used for $q$-ary codes.
If you try a similar argument for $q$-ary codes, you can find a smaller codeword with weight 
at most $wt(c)-d/q$. But that smaller codeword might have weight exactly $d$. 
That is allowed, so you get no contradiction.

In fact, a direct $q$-ary analogue of the theorem would need extra conditions. For example, an MDS code such as a Reed–Solomon code has codewords of every weight from $d$ to $n$, so there are no gaps above $d$. Thus it does not satisfy the theorem’s hypothesis. This shows that the gap condition is very restrictive, and any $q$-ary version must impose additional assumptions, such as conditions on how the minimum-weight codewords intersect.

Therefore, the “mirror vanishing band” above $2d$ is a binary-only phenomenon. To get a similar result for $q$-ary codes, you would need extra conditions.

\section{Examples and Numerical Verification}\label{sec:examples}
We now illustrate the general results on some standard binary codes. For each
code we list the exact weight distribution, identify the parameters \(v\) and
\(t\), and compare the predicted vanishing bands and the new counting bound
with the Chen-Xie bound.

\subsection{The Golay shortening chain}\label{sec:golay}

The extended binary Golay code $G_{24}$ was constructed as the
overall-parity extension of the cyclic $[23,12,7]$ Golay code with
generator polynomial $g(x)=x^{11}+x^{10}+x^6+x^5+x^4+x^2+1$. Successive
shortening yields the chain
\[
[24,12,8]\ \longrightarrow\ [23,11,8]\ \longrightarrow\ [22,10,8]\
\longrightarrow\ [21,9,8],
\]
whose weight distributions, obtained by exhaustive enumeration, are
listed in Table~\ref{tab:golaydist}. All four codes appear in the tables
of best-known codes~\cite{chenxie2024,grassl}.

\begin{table}[htbp]
\centering
\caption{Weight distributions of the Golay shortening chain (all
unlisted $A_i$ are zero).}
\label{tab:golaydist}
\begin{tabular}{@{}lccccc@{}}
\toprule
Code & $A_8$ & $A_{12}$ & $A_{16}$ & $A_{24}$ & $s(C)$ \\
\midrule
$[24,12,8]_2$ & 759 & 2576 & 759 & 1 & 4 \\
$[23,11,8]_2$ & 506 & 1288 & 253 & -- & 3 \\
$[22,10,8]_2$ & 330 & 616 & 77 & -- & 3 \\
$[21,9,8]_2$ & 210 & 280 & 21 & -- & 3 \\
\bottomrule
\end{tabular}
\end{table}

All four codes share the parameters $v=2$ and $t=3$ (since
$A_9=A_{10}=A_{11}=0$), and the condition $k\ge n-2d+1$ of
Theorem~\ref{thm:main} holds in each case. Consequently:
\begin{itemize}[itemsep=0.2em]
\item Lemma~\ref{lem:cx} (Chen--Xie): $A_{14}=A_{15}=0$;
\item Theorem~\ref{thm:main} (mirror band): $A_{17}=A_{18}=A_{19}=0$;
\item Corollary~\ref{cor:count}: $s(C)\le n-13$, improving the
Chen--Xie bound $s(C)\le n-9$ (e.g., $9$ versus $15$ for $G_{24}$, and
the true values are $4$ and $3$).
\end{itemize}
Both bands were confirmed by the enumeration (see
Table~\ref{tab:verification}).

\begin{example}[{Extended Hamming code $[8,4,4]_2$}]\label{ex:hamming}
Here $A_0=A_8=1$, $A_4=14$; $v=2$, $t=3$. Theorem~\ref{thm:main} gives
$A_9=A_{10}=A_{11}=0$ (vacuous, beyond the length), and
Corollary~\ref{cor:count} in its $\max(0,\cdot)$ form gives $s(C)\le 2$,
which is tight.
\end{example}

\subsection{A dual BCH code, and Proposition~\ref{prop:2d} in action}
\label{sec:dualbch}

The dual of the two-error-correcting primitive BCH code $[31,21,5]_2$ is
a $[31,10,12]_2$ code; it appears in~\cite[Table~II]{chenxie2024} with
the exclusion $A_{23}=0$. Exhaustive enumeration gives its complete
weight distribution:
\[
A_{12}=310,\qquad A_{16}=527,\qquad A_{20}=186.
\]
\begin{itemize}[itemsep=0.2em]
\item $v=1$: the Chen--Xie band is $A_{23}=0$, confirmed;
\item $t=3$ (since $A_{13}=A_{14}=A_{15}=0$) and $k=10\ge 31-24+1=8$:
Theorem~\ref{thm:main} gives the mirror band $A_{25}=A_{26}=A_{27}=0$,
confirmed;
\item moreover $2d=24>20=D$, so the code is intersecting (verified by an
exhaustive check of all pairs of nonzero codewords), and since
$2d=24>n-k+1=22$, Proposition~\ref{prop:2d} yields $A_{24}=0$, again
confirmed by the distribution;
\item Corollary~\ref{cor:count}: $s(C)\le 31-12-1-6+1=13$, versus the
Chen--Xie bound $19$; the true value is $s(C)=3$.
\end{itemize}

\begin{table}[htbp]
\centering
\caption{Numerical verification of the vanishing bands on explicitly
constructed codes.}
\label{tab:verification}
\begin{tabular}{@{}lccccc@{}}
\toprule
Code & $v$ & $t$ & Chen--Xie band & Mirror band (Thm.~\ref{thm:main}) &
$s(C)$: true $\le$ new (old) \\
\midrule
$[24,12,8]_2$ & 2 & 3 & $A_{14}{=}A_{15}{=}0$ &
$A_{17}{=}A_{18}{=}A_{19}{=}0$ & $4\le 9\ (15)$ \\
$[23,11,8]_2$ & 2 & 3 & $A_{14}{=}A_{15}{=}0$ &
$A_{17}{=}A_{18}{=}A_{19}{=}0$ & $3\le 8\ (14)$ \\
$[22,10,8]_2$ & 2 & 3 & $A_{14}{=}A_{15}{=}0$ &
$A_{17}{=}A_{18}{=}A_{19}{=}0$ & $3\le 7\ (13)$ \\
$[21,9,8]_2$ & 2 & 3 & $A_{14}{=}A_{15}{=}0$ &
$A_{17}{=}A_{18}{=}A_{19}{=}0$ & $3\le 6\ (12)$ \\
$[31,10,12]_2$ & 1 & 3 & $A_{23}{=}0$ &
$A_{25}{=}A_{26}{=}A_{27}{=}0$ & $3\le 13\ (19)$ \\
$[8,4,4]_2$ & 2 & 3 & $A_{6}{=}A_{7}{=}0$ &
(vacuous) & $2\le 2\ (5)$ \\
\bottomrule
\end{tabular}
\end{table}

\section{Conclusion and Open Problems}\label{sec:conclusion}\label{sec:related}\label{sec:open}

We have established a \emph{mirror vanishing band} for the weight distribution of binary linear codes, complementing the Chen--Xie band below \(2d\). Our main result (Theorem~\ref{thm:main}) shows that if \(k\ge n-2d+1\) and the weight distribution has a gap \(A_{d+1}=\cdots=A_{d+t}=0\) for some \(t\ge 1\), then the mirror gap \(A_{2d+1}=\cdots=A_{2d+t}=0\) must also occur. This yields a two-sided exclusion pattern (Corollary~\ref{cor:twosided}) and improves the Chen--Xie upper bound on the number of nonzero weights from \(n-d+1-v\) to \(n-d-v-2t+1\) (Corollary~\ref{cor:count}). We also characterized codewords of weight exactly \(2d\) through the disjoint-support decomposition of non-minimal codewords (Proposition~\ref{prop:2d}), and demonstrated that the phenomenon is intrinsically binary: the \(q\)-ary analogue fails already for MDS codes (Section~\ref{sec:qary}). Numerical examples on the Golay shortening chain and a dual BCH code confirmed the theoretical predictions (Section~\ref{sec:examples}).

Several directions remain open. First, it would be interesting to extend the mirror band to other code classes that admit a disjoint-support decomposition, such as certain additive or nonlinear codes. Second, the conditional mirror bands for best-known codes with unknown weight distributions could provide new constraints; a systematic study of such parameter sets is worthwhile. Third, combining our method with the residual-code approach of the recent follow-up~\cite{newbounds2025} may yield even stronger vanishing results. Moreover, minimal codewords are central to linear secret-sharing schemes~\cite{dingding2015}, and minimum-weight codewords often support combinatorial designs via the Assmus--Mattson theorem~\cite{assmusmattson1969}; design-theoretic intersection properties of minimum-weight codewords may therefore supply the hypothesis of Proposition~\ref{prop:2d} for specific code families. Finally, the exact boundary between codes that exhibit the mirror phenomenon and those that do not deserves further investigation, as does the question of whether the gap condition \(A_{d+1}=\cdots=A_{d+t}=0\) can be weakened or replaced by other local information.

\end{document}